\pdfoutput=1
\documentclass[pra,aps,twocolumn,twoside,nofootinbib,showkeys,showpacs]{revtex4}
\usepackage{float}
\usepackage{indentfirst}
\usepackage{url}
\usepackage[stable]{footmisc}
\usepackage{fancyhdr}
\usepackage{appendix}
\usepackage{epstopdf}
\usepackage[all]{xy}
\usepackage{ifpdf}
\usepackage{color}
\ifpdf
 \usepackage[pdftex]{graphicx}
 \else
\usepackage{graphicx}
 \fi
\usepackage{graphicx}
\usepackage{youngtab}
\usepackage[T1]{fontenc}
\usepackage{textcomp}
\usepackage{newcent}
\usepackage[sc,noBBpl]{mathpazo}
\usepackage{amsmath,amsfonts,amssymb,amsthm}
\usepackage[mathscr]{eucal}
\theoremstyle{plain}
\newtheorem{theorem}{Theorem}
\newtheorem{lemma}[theorem]{Lemma}

\newtheorem{corollary}[theorem]{Corollary}
\newtheorem{conjecture}[theorem]{Conjecture}
\newtheorem{definition}[theorem]{Definition}
\newtheorem{fact}[theorem]{Fact}
\newtheorem{notation}[theorem]{Notation}
\newtheoremstyle{note}{\topsep}{\topsep}{\slshape}{}{\scshape}{}{ }{}
\theoremstyle{note}

\newtheorem{remark}[theorem]{Remark}
\numberwithin{equation}{section}
\newcommand\field{\mathbb}

\newcommand\II{\field{I}}

\newcommand\tr{\operatorname{Tr}}

\renewcommand{\>}{\rangle}

\newcommand\be{\begin{equation}}
\newcommand\ee{\end{equation}}
\newcommand\bea{\begin{array}}
\newcommand\eea{\end{array}}
\newcommand\ben{\begin{eqnarray}}
\newcommand\een{\end{eqnarray}}
\newcommand\ot{\otimes}

\newcommand\bei{\begin{itemize}}
\newcommand\eei{\end{itemize}}
\newcommand\bee{\begin{enumerate}}
\newcommand\eee{\end{enumerate}}

\newcommand{\poly}{\operatorname{poly}}

\begin{document}
\title{Local random quantum circuits are approximate polynomial-designs - numerical results}

\author{Piotr \'Cwikli\'nski$^{1}$,
Micha{\l} Horodecki$^{2,4}$, Marek Mozrzymas$^{3}$, \L{}ukasz Pankowski$^{2,4}$, and Micha{\l} Studzi\'nski$^{2,4}$}
\affiliation{
$^1$ Institute for Quantum Information, RWTH Aachen University, D-52056 Aachen, Germany \\
$^2$ Institute of Theoretical Physics and Astrophysics, University of Gda\'nsk, 80-952 Gda\'nsk, Poland \\
$^3$ Institute for Theoretical Physics, University of Wroc\l{}aw, 50-204 Wroc\l{}aw, Poland \\
$^4$ National Quantum Information Centre of Gda\'nsk, 81-824 Sopot, Poland
}

\date{\today}
%\date{1 July 2009}

\begin{abstract}
We numerically investigate the statement that local random quantum circuits acting on $n$ qubits composed of polynomially many nearest neighbour two-qubit gates form an approximate unitary $\poly(n)$-design [F.G.S.L. Brand\~ao et al., arXiv:1208.0692]. Using a group theory formalism, spectral gaps that give a ratio of convergence to a given $t$-design are evaluated for a different number of qubits $n$ (up to $20$) and degrees $t$ ($t=2,3,4$ and $5$), improving previously known results for $n=2$ in the case of $t=2$ and $3$. Their values lead to a conclusion that the previously used lower bound that bounds spectral gaps values may give very little information about the real situation and in most cases, only tells that a gap is closed. We compare our results to the another lower bounding technique, again showing that its results may not be tight. 
%This implies that the size of a circuit that forms an $\varepsilon$-approximate $t$-design, calculated according to the bounds, in principle, is very different from %that predicted according to the convergence of the values of spectral gaps.
\end{abstract}


\maketitle

\section{Introduction}
\label{sec:introduction}
Random unitary matrices have established their place as a useful and powerful tool in theory of quantum information and computation. For example, they are used in the encoding protocols for sending information
down a quantum channel \cite{Abeyesinghe2009-Haar}, approximate encryption of quantum information \cite{Hayden2004-Haar}, quantum datahiding \cite{Hayden2004-Haar}, information locking \cite{Hayden2004-Haar}, process tomography \cite{Emerson2005-unitary}, state distinguishability \cite{Sen-Haar}, and equilibration of quantum states \cite{Brandao2011-thermo, Masanes2011-thermo,Vinayak2011-thermo,Cramer2011-thermo} or some other problems in foundation of statistical mechanics \cite{Zanardi-stat2010}.

But their is a problem with them - they are not very favorable from a computational point of view. Why? The answer is the following - to implement a random Haar unitary one needs an exponential number of two-qubit gates and random bits (in other words, to sample from the Haar measure with error $\varepsilon$, one needs $exp(4^n \log(\frac{1}{\varepsilon}))$ different unitaries). To omit this problem, one can construct, so-called, approximate random unitaries or pseudo-random unitaries. Using them, an efficient implementation is possible.

An approximate unitary $t$-design is a distribution of unitaries which mimic properties of the
Haar measure for polynomials of degree up to $t$ (in the entries of the unitaries) \cite{Dankert2009-Haar, Gross2007-Haar, Emerson2005-Haar, Oliveira2007-Haar, Dahlsten2007-Haar, Harrow2009-tdesign, Diniz2011-tdesign, Arnaud2008-tdesign, Znidaric2008-tdesign, Roy2009-tdesign, Brown2010-tdesign, Brandao2010-tdesign, Hayden2007-tdesign, Harrow2009a-tdesign}. Approximate designs have a number of interesting applications in quantum information theory replacing
the use of truly random unitaries (see e.g. \cite{Oliveira2007-Haar, Dahlsten2007-Haar, Hayden2007-tdesign, Brandao2010-tdesign, Emerson2003-tdesign, Low2009-tdesign}). What is more, other particular constructions of approximate unitary $t$-designs together with some applications in quantum physics, have been formulated, let us mention here, for example, a recent construction of diagonal unitary $t$-designs \cite{Nakata2012-tdesigns1, Nakata2012-tdesigns2}.

In 2009, Harrow and Low \cite{Harrow2009-tdesign} stated a conjecture that polynomial sized random quantum circuits acting on $n$ qubits
form an approximate unitary $\poly(n)$-design. To give an example supporting their statement, the authors, also in 2009, presented efficient constructions of quantum $t$-designs, using a polynomial number
of quantum gates and random bits, for $t = O(\log(n))$ \cite{Harrow2009a-tdesign}.

However, it took some time to verify the conjecture from \cite{Harrow2009-tdesign}. At that time, it was already known that there exist efficient approximate unitary $2$-designs in $\mathbb{U}(d^n)$, where efficient means that unitaries are created by polynomial (in $n$) number of two-qubit gates and the distribution of unitares can also be sampled in polynomial time (in other words that random circuits are approximate unitary $2$-design) \cite{Oliveira2007-Haar, Dahlsten2007-Haar, Harrow2009-tdesign, Diniz2011-tdesign, Arnaud2008-tdesign, Znidaric2008-tdesign, DiVincenzo-Clifford}. In 2010, Brand\~ao and Horodecki \cite{Brandao2010-tdesign}, made one step further proving that polynomial random quantum circuits are approximate unitary $3$-designs. Quite recently, a break-through has been made for the above problem. In \cite{Brandao-Harrow-Horodecki}, authors has proved that local random quantum circuits acting on $n$ qubits composed of polynomially many nearest neighbor two-qubit gates form an
approximate unitary $\poly(n)$-design, settling the conjecture from \cite{Harrow2009-tdesign} to be affirmative. Their proof is based on
techniques from many-body physics, representation theory and combinatorics. In particular, one of the tool used to obtain the main result was
estimation of the spectral gap of the frustration-free local quantum Hamiltonian that can be used to study the problem (instead of a quantum
circuit).

In this paper, we analyze two aspects of their statement: First, we numerically verify and investigate it, calculating spectral gaps for
the increasing number of qubits (in circuit) $n$ (up to $n=20$, in the best case) and degree $t$ ($t=2,3,4,5$). Previously, exact values were known only for
$n=2$ and $t=2,3$. This may be of independent interest since, in many-body physics, the knowledge of spectral gaps of local Hamiltonians is
useful in studying many-body systems (see, for example, \cite{Schwarz2012-PEPS, Hastings2007-area, Nach2006-cluster, Brandao2011-thermo}). In addition, in \cite{Brandao2010-tdesign}, the authors obtained that for $n=2$, gap values for $t=2$ and $3$ are the same; we obtain the matching of all calculated spectral gap values in the case of $t=2$ and $3$, which is a unique feature for this exact degrees $t$. Second, in \cite{Brandao-Harrow-Horodecki}, in order to proof
the main statement, a lower bound for spectral gaps was derived (independent from the number of qubits).
However, based on our results, it can be concluded that
lower bounding of spectral gaps may not be tight. We show that there could be a large difference in actual values of spectral gaps and corresponding
lower bounds. What is more, our calculations for increasing values of $n$ and $t$ lead to a conclusion that first, the lower bound from \cite{Brandao-Harrow-Horodecki} is hard to estimate, second, it gives only the limited information about the actual value of spectral gaps (that a gap is closed).
Simultaneously, we also compare our results to the another lower bound, derived in \cite{Knabe88}, and show that predictions obtained according to it can be better than that from \cite{Brandao-Harrow-Horodecki}. Nevertheless, also in this case, predictions
according to that bound can be inconclusive in some cases, showing that, in principle, obtaining a good bound is a demanding task.

Our paper in structured as follows. In Section \ref{Sec:2} we start from introducing local (random) quantum circuits which, using the formalism of superoperators, we connect with approximate unitary $t$-designs. In Section \ref{Sec:3} we recall the statement that local quantum circuits of a given length form approximate unitary $t$-designs. We also recall that to verify and check this connection, calculations of spectral gaps and second largest eigenvalues of local Hamiltonians connected to the problem are necessary. Section \ref{Sec:num} is the key section of this paper. We first show, how to connect our problem of checking the connection between local quantum circuits and $t$-designs (calculating the spectral gaps) with symmetric groups
$S(t)$, where $t$ corresponds to a degree in $t$-design. Then, we present our numerical calculations for spectral gaps for different number of qubits $n$ and different degrees $t$ which we later compare to results obtained using the techniques for lower bounding the spectral gaps.

\section{Random unitary circuits and approximate $t$-designs}
\label{Sec:2}
In this section, we present the formalism of local random quantum circuits and approximate unitary designs. We want to point here, that Sec. \ref{Sec:2} and \ref{Sec:3} are mainly based on \cite{Brandao2010-tdesign, Brandao-Harrow-Horodecki}, so for the full analysis (proofs, etc.), we refer to these papers.

We consider $n$ qubits (so, from now $d=2$), and apply $l$ steps of a random circuit (random walks on $\mathbb{U}(d^n)$.
\begin{definition}
\textbf{(Local quantum circuit)} In each step of the walk an index $i$ is chosen uniformly at random from $[n]$ and then a unitary $U_{i, i+1}$ drawn from the Haar measure on $\mathbb{U}(d^2)$ is applied to the two neighboring qudits $i$ and $i+1$.
\end{definition}

There are several different definitions of $\varepsilon-$approximate unitary $t$-designs \cite{Low2010-PhD} from which let us mention the following.
\begin{definition}
\label{def2}
\textbf{(Approximate unitary t-design)} Let $\left\{ \mu, U \right\}$ be an ensemble of unitary operators from $\mathbb{U}(2^n)$. Define
\be {\cal G}_{\mu, t}(\rho) = \int_{\mathbb{U}(d)} U^{\ot t} \rho (U^{\dag})^{\ot t} \mu(dU) \ee
and
\be {\cal G}_{H, t}(\rho) = \int_{\mathbb{U}(d)} U^{\ot t} \rho (U^{\dag})^{\ot t} \mu_H(dU), \ee
where $\mu_H$ is the Haar measure. Then the ensemble is a $\varepsilon$-approximate unitary $t$-design if
\be ||{\cal G}_{\mu, t} - {\cal G}_{H, t}||_{2 \rightarrow 2} \leq \varepsilon, \ee
where the induced Schatten norm
$||\Lambda(\operatorname{X})||_{p\rightarrow q}=\mathop{\operatorname{sup}}\limits_{\operatorname{X}\neq 0}\frac{||\Lambda(\operatorname{X})||_p}{||\operatorname{X}||_q}$ is used.
\end{definition}

%Having establish these, we can ask the following questions
%\begin{itemize}
%	\item How well a random circuit can mimic a random global unitary? (i.e.
%how fast the related random walk converges to the Haar measure)
%	\item How well a random circuit can mimic 'twirling' $\int dU U \ot U \ot \ldots \ot U (\bullet) U^{\dag} \ot U^{\dag} \ot \ldots \ot U^{\dag}$? (i.e. how fast the %random walk converges to so-called $t$-design)
%	\item More generally: How fast a random circuit converges to $t$-designs?
%\end{itemize}
%Recently, in a series of papers, an effort has been made to answer them \cite{Brandao2010-tdesign, Brandao-Harrow-Horodecki}.

\section{Local random quantum circuits are approximate polynomial-designs}
\label{Sec:3}
In this section we will review some basis facts about local random circuits. At the end, we recall the statement that local quantum circuits of a given length form approximate unitary $t$-designs.
%First, let us explore some connections between superoperators and operators. For a superoperator $\bG$ given by
%\be {\cal G}(X) = \sum_{k} A_k X B_{k}^{\dag}, \ee
%where $\dag$ denotes the Hermitian conjugate; we define the operator
%\be G = \sum_k A_k \ot B_{k}^{\star}, \label{G} \ee
%with $\star$ being the complex conjugate.

%Let $X$ be a normalized operator, such that $\tr(X X^{\dag}) = 1$. What is more, assume that ${\cal G}(X) = \lambda X$, for a complex eigenvalue $\lambda$, i.e., $X$ %is an eigenoperator of ${\cal G}$ with eigenvalue $\lambda$. Then defining $|X\> = X \ot \II |\phi\>$, with
%\be |\phi\> = \sum_k |k\> \ot |k\>, \ee
%it holds that $G|X\> = \lambda |X\>$, i.e., $|X\>$ is an eigenvector of $G$ with eigenvalue $\lambda$.

%A direct implication of this correspondence is that
%\be ||{\cal G}||_{2 \rightarrow 2} = ||G||_{\infty}, \ee
%where the norm $2 \rightarrow 2$ is the norm, defined as in .
Let $\mu$ be a measure on $\mathbb{U}(2^n)$ induced by one step of the local random circuit model and $\mu^{\star l}$ measure induced by $l$ steps on such a model, then one can show that (having in mind that for a superoperator $\cal G$ and an operator $G$ that have the same set of eigenvalues holds $||{\cal G}||_{2 \rightarrow 2} = ||G||_{\infty}$, see Appendix \ref{App:C})

\begin{theorem} \cite{Brandao-Harrow-Horodecki} \label{th}
\be ||G_{\mu^{\star l},t}-G_{\mu_H,t}||_{\infty} = \lambda_2(\int_{\mathbb{U}(d)} U^{\ot t} \ot (U^{\star})^{\ot t} \mu(dU))^l \label{spec} \ee
where $G_{\mu^{\star l}, t} = \int_{\mathbb{U}(d)} U^{\ot t} \ot (U^{\star})^{\ot t} \mu^{\star l}(dU)$ and $\lambda_2$ stands for the second largest eigenvalue of $G_{\mu,t}$. But
$\mu = \frac{1}{n}\sum_{i=1}^{n-1} \mu_{H}(i, i+1)$, so
\be \lambda_2(\int_{\mathbb{U}(d)} U^{\ot t} \ot (U^{\star})^{\ot t} \mu(dU)) = 1 -  \frac{\Delta(H_{n, t})}{n}, \ee
with $H_{n, t} = \sum_{i=1}^{n-1} h_{i,i+1}$, with local terms $h_{i,i+1} = I - \int_{\mathbb{U}(d)} U_{i,i+1}^{\ot t} \ot (U_{i,i+1}^{\star})^{\ot t} \mu_{H}(dU)$ and $\Delta(H_{n, t})$ the spectral gap of the local Hamiltonian $H_{n,t}$ (see, Fig. \ref{fig:H}).
\end{theorem}

\begin{figure}[h]
	\includegraphics[width=0.45\textwidth]{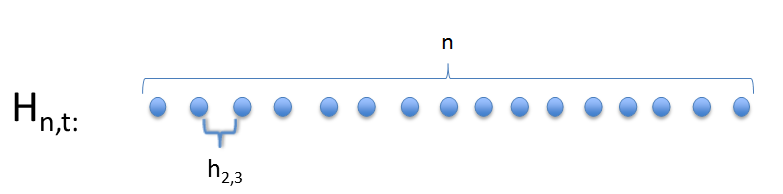}
	\label{fig:H}
	\caption{\textbf{(Color online) A local Hamiltonian acting on $n$ subsystems $H_{n, t} = \sum_{i=1}^{n-1} h_{i,i+1}$, each of dimension $d^{2t}$. Its spectral gap value $\Delta(H_{n, t})$ can be connected with the second largest eigenvalue of the operator $G_{\mu,t}$.}}
\end{figure}

After a successful estimation of the spectral gap from Eq. \eqref{spec}, one can show that
\begin{theorem} \cite{Brandao-Harrow-Horodecki} \label{MH}
Local random circuits of size $(\log\left(\frac{1}{\varepsilon}\right)2\log(d)\log(t))t^5n^2$ form an $\varepsilon$-approximate unitary $t$-design.
\end{theorem}

Thus the problem reduces to analysis of spectral gap of the operator $G_{\mu,t}$. Now, it is important to ask the following:
\begin{itemize}
	\item How a spectral gap depends on number of qubits $n$ and the degree of design $t$?
\end{itemize}
In the next section, an answer to this question is provided.

\section{Spectral gaps - numerical results}
\label{Sec:num}
In this section, we present our numerical results for spectral gaps for different degrees $t$ in  $t$-designs (for simplicity, we gather all results in Table \ref{tab-gap}). What is more, we (where it was possible) compare our results with two lower bounds for spectral gaps: $1)$ The "local" lower bound obtained in \cite{Knabe88}; $2)$ The "global" lower bound derived in \cite{Brandao-Harrow-Horodecki}. To evaluate the spectral gaps, the $Matlab$ software and a $C++$ code have been used.

\subsection{Local Hamiltonians as a tool for calculating spectral gaps}
We already showed that calculations of spectral gaps can be connected with the second largest eigenvalue of $G_{\mu,t}$. Now, we will show how to connect calculations of spectral gaps (equivalently, second largest eigenvalues) with symmetric groups $S(t)$, where $t$ plays a role of the degree of $t$-design; using techniques for local Hamiltonians introduced in Sec. \ref{Sec:3}.

At the beginning, let us remind that our Hamiltonian is of the form
$H_{n, t} = \sum_{i=1}^{n-1} h_{i,i+1}$, with local terms $h_{i,i+1} = I - P_{i,i+1}$, and the notation as in Sec. \ref{Sec:3}

Let us consider superoperators associated with projectors $P_{i,i+1}$
\be
\label{eq:Pii}
\mathcal{P}_{i,i+1}(X) = \int_{\mathbb{U}(d)} U_{i,i+1}^{\ot t} X (U_{i,i+1}^{\star})^{\ot t} \mu_{H}(dU)
\ee

Now, we can find, as a consequence of the Schur-Weyl duality~\cite{Weyl} that all operators $X$ invariant under action of $\mathcal{P}_{i,i+1}$ can be written
as a sum of permutation operators $V_{\pi,(i,i+1)}$, permuting $t$ copies of the Hilbert space $\mathcal{H}^{\ot t}_{i,i+1}$.
We know that $V_{\pi,(i,i+1)}=V_{\pi,i} \ot V_{\pi,i+1}$, where $V_{\pi,i}$ is the operator representing some permutation $\pi \in S(t)$ acting on $H_i^{\ot t}$. Hence $P_{i,i+1} \subset P_i \ot P_{i+1}$, where operators $P_i$ are given by the expression
\be
\label{Pi}
P_i=\int_{\mathbb{U}(d)} U_{i}^{\ot t}\ot (U_{i}^{\star})^{\ot t} \mu_{H}(dU).
\ee
Thanks to the above consideration we can deduce that operator $X$ can be identify with the operator $G_{\mu,t}$ from Theorem \ref{th} and written according to the formula
\be
\label{x}
\begin{split}
G_{\mu,t} &= P_{1,2} \ot I_3 \ot \ldots \ot I_n + I_1 \ot P_{2,3} \ot I_4 \ot \ldots \ot I_n + \\
&+\ldots + I_1 \ot I_2 \ot \ldots \ot P_{n-1,n}.
\end{split}
\ee
In a subspace spanned by permutation operators acting on Hilbert space $\mathcal{H}^{\ot t}$ we are able to construct operator basis which is orthogonal in the Hilbert-Schmidt scalar product (see Appendix \ref{App:A}). Using this basis we can calculate two biggest eigenvalues of operators $X$ which are necessary to know the spectral gap.

\subsection{Lower bounding a spectral gap}
Herewith, we present two methods for lower bounding the spectral gaps that we will use later to compare to values of spectral gaps obtained numerically.

1. "Global" bound:
\begin{lemma} \cite{Brandao-Harrow-Horodecki}  For every integers $n$ (number of qubits) and $t$ ($t$ in $t$-design), with $n \geq \left\lceil 10 \log(t)\right\rceil$, the spectral gap $\Delta(H_{n,t})$, can be lower bounded as follows:
\be \label{boundMH}
\Delta(H_{n,t}) \geq \frac{\Delta(H_{\left\lceil 2 (\log(d))^{-1} \log(t)\right\rceil,t})}{8(\log(d))^{-1}\log(t)},
\ee
with $d$ being the dimension of the Hilbert space and $\left\lceil a \right\rceil$ denotes a smallest integer $k$ satisfying $k \geq a$.
\end{lemma}

2. "Local" bound:
\begin{lemma} \cite{Knabe88} For every integers $n$ and $t$, the spectral gap $\Delta(H_{n,t})$ of local Hamiltonian can be lower bounded in the following way:
\be \label{boundKnabe}
\Delta(H_{n,t}) \geq \frac{k \Delta(H_{k,t}) - 1}{k - 1},
\ee
where $\Delta(H_{k,t})$ is the Hamiltonian restricted to $k$ qubits: $H_k = \sum_{i=1}^{k} h_{i,i+1}$.
\end{lemma}

Let us explain now, why we used the terms "global" and "local" to describe these two bounds. From Eq. \eqref{boundMH} it can be noticed that to lower bound a given spectral gap, one needs to compute only one spectral gap, the one that is given by $t$ and the dimension of the system (but we consider qubits only, so the value $d = 2$ is set). The bound is global in the sense that it does not change with the number of qubits - it remains constant for an arbitrary length of a quantum circuit. On the contrary, the second bound (from Eq. \eqref{boundKnabe}) has a totally reverse property. There, the gap for a higher number of qubits implies better precision of the bound.

Of course, as we will see later, both bound have pros and cons. For example, the "global" bound usually bounds the spectral gap value in a harsh way (only giving an info that a gap is open) and, what is more, one can observe that for big values of $t$, to estimate that bound, one needs to know a value of the spectral gap for a big number of qubits. Keeping in mind that calculating spectral gaps is, in principle, computationally hard problem, the effectiveness of the "global" bound is limited. The good thing about it is that it is always positive (giving thus nonzero convergence rate
of random circuits to a given $t$-design). On the other hand, it is not true for the second bound. For small numbers of qubits, the "local" bound can be negative and it means that one needs to increase its precision by calculating the lower bound using a gap for a bigger number of qubits. Moreover, the value of the bound sometimes fluctuates - so the value of the bound for a number of qubits $k$ can be possibly worse than that for $k-1$ qubits. The advantage of this bound is that when it is positive its value is usually closer to the exact value than this predicted by the bound from Eq. \eqref{boundMH}.

\subsection{Method of Calculations}
In this section we will present some methods of calculations used in this paper. Notation is mostly taken from~\cite{Brandao2010-tdesign}.\\

We know that for an arbitrary subspace of  operator space we are able to construct basis of operators which is orthonormal in the Hilbert-Schmidt scalar product. For this construction we use linear combination of nonorthogonal permutation operators acting on $\mathcal{H}^{\ot t}$:
\be
R_k=\sum_{\pi \in S(t)}b_{k \pi}(t)V_{\pi}.
\ee
Using operators $R_k$ we can rewrite operators $P_{i,i+1}$ and $P_i$ from Eq.~\eqref{eq:Pii} and \eqref{Pi} in a form
\be
P_{i,i+1}=\sum_k|R_k^{(i,i+1)}\>\<R_k^{(i,i+1)}|, \quad P_i=\sum_k|R_k^{(i)}\>\<R_k^{(i)}|.
\ee
To illustrate the action of operators $R_k^{(i,i+1)}$ consider a $t \times n$ lattice (see Figure~\ref{fig:0}), then $R_k^{(i,i+1)}$ acts jointly on systems from $i^{\text{th}}$ and $(i+1)^{\text{th}}$ column.
\begin{figure}[ht]
\includegraphics[width=0.3\textwidth, height=0.2\textheight]{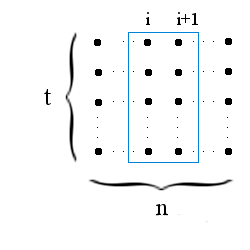}
\caption{\textbf{(Color online) The figure presents a $t \times n$ lattice. We have $n$ columns and in every column there are $t$ qubits on which permutations from $S(t)$ act. In a blue rectangle, we marked $i^{\text{th}}$ and $(i+1)^{\text{th}}$ column on which the operator $R_k^{(i,i+1)}$ acts.}
\label{fig:0}}
\end{figure}
Finally to obtain a representation of an operator $X$  from equation~\eqref{x} in our product basis i.e.
\be
X=\sum_{\bf{kl}} a_{\bf{k}\bf{l}}|R_{k_1}^{(1)}\ot \cdots \ot R_{k_n}^{(n)}\>\<R_{l_1}^{(1)}\ot \cdots \ot R_{l_n}^{(n)}|,
\ee
where ${\bf{k}}=(k_1\ldots k_n)$ and ${\bf{l}}=(l_1\ldots l_n)$ are multiindices, we have to express operators $P_{i,i+1}$  in terms of a product basis.
For this purpose we use argumentation from~\cite{Brandao2010-tdesign}, having:
\be
\label{double}
R^{(1,2)}_k=\sum_{s,u}r_{su}^{(k)}R_s^{(1)}\otimes R_u^{(2)},
\ee
where $r_{su}^{(k)}$ are some coefficients which we want to know. In this paper, to calculate numbers $r_{su}^{(k)}$ we use the Schur basis (see, for example,~\cite{Boerner}). Then every operator $R_k^{(1,2)}$ corresponds to a linear combination of $E_{kl}^{\alpha} \ot E_{mn}^{\beta}$, where operators $E_{ij}^{\alpha}$ form an operator basis in a given invariant subspace of $\mathcal{H}^{\ot t}$ labeled by $\alpha$. Now we see that an index $k$ in Eq.~\eqref{double} is indeed a multiindex.
The general method of constructing such an operator basis via representation theory is given in Appendix~\ref{App:A}. Of course computing eigenvalues in an arbitrary basis (also in our basis) is quite hard, because complexity of calculations grows very fast with the parameter $t$.

\subsection{Numerical results}
\begin{table}[htp]
	\begin{center}
		\begin{tabular}[t]{|r|r|r|}
		\multicolumn{1}{c}{$t$} & \multicolumn{1}{c}{$k$} & \multicolumn{1}{c}{$\Delta(H_{k, t})$} \\ \hline
		2,3 & 2 & 0.6  \\ \hline
		2,3 & 3 & 0.43431  \\ \hline
		2,3 & 4 & 0.35279  \\ \hline
		2,3 & 5 & 0.30718  \\ \hline
		2,3 & 6 & 0.27922  \\ \hline
		2,3 & 7 & 0.2609  \\ \hline
		2,3 & 8 & 0.24825  \\ \hline
		2,3 & 9 & 0.23915  \\ \hline
		2,3 & 10 & 0.23241  \\ \hline
		  2 & 11 & 0.2273  \\ \hline
		  2 & 12 & 0.2232  \\ \hline	
		  2 & 13 & 0.2201  \\ \hline
			2 & 14 & 0.2175  \\ \hline
			2 & 15 & 0.2154  \\ \hline
			2 & 16 & 0.2136  \\ \hline
		  2 & 17 & 0.2122  \\ \hline		
		  2 & 18 & 0.2109  \\ \hline
		  2 & 19 & 0.2098  \\ \hline
		  2 & 20 & 0.2089  \\ \hline		
		\multicolumn{1}{c}{} & \multicolumn{1}{c}{} & \multicolumn{1}{c}{} \\
		\end{tabular}
		\hspace{0.1in}
		\begin{tabular}[t]{|r|r|r|}
		\multicolumn{1}{c}{$t$} & \multicolumn{1}{c}{$k$} & \multicolumn{1}{c}{$\Delta(H_{k, t})$} \\ \hline
		4 & 2 & 0.5  \\ \hline
		4 & 3 & 0.45298644403  \\ \hline
		4 & 4 & 0.42486035753  \\ \hline
		4 & 5 & 0.41022855573  \\ \hline
		\multicolumn{1}{c}{} & \multicolumn{1}{c}{} & \multicolumn{1}{c}{}
		\end{tabular}
		\hspace{0.1in}
		\begin{tabular}[t]{|r|r|r|}
		\multicolumn{1}{c}{$t$} & \multicolumn{1}{c}{$k$} & \multicolumn{1}{c}{$\Delta(H_{k, t})$} \\ \hline
		5 & 2 & 0.37373469602  \\ \hline
		5 & 3 & 0.32912548483  \\ \hline
		\multicolumn{1}{c}{} & \multicolumn{1}{c}{} & \multicolumn{1}{c}{}
		\end{tabular}
		
	\caption{Numerically calculated spectral gaps for different degrees $t$ and qubits number $k$.}
	\label{tab-gap}
	\end{center}
\end{table}

Here, we present our numerically calculated values of spectral gaps for different cases, together with corresponding lower bounds (both "`local"' and "`global"' when possible). In Figure \ref{fig:1}, there are spectral gaps $\Delta(H_{k, t})$ for the symmetric group $S(2)$ and different number of qubits $n$, corresponding to the $2$-design. Both lower bounds are also marked. Similarly, values of spectral gaps $\Delta(H_{k, t})$, for $t=3,4,5$ ($3,4,5$-designs) can be found in Figures \ref{fig:2}, \ref{fig:3} and \ref{fig:4}. Notice, that with increasing $t$, the "`global"' bound tends to less and less information, namely, telling us that gaps are closed. What is more, for the last plot (Figure \ref{fig:4}), the "local" bound does not give any bound since it takes a negative value. Knowledge of gaps for $k \geq 4$ is required to resolve this problem.
\begin{figure}[ht]
\includegraphics[width=0.45\textwidth, height=0.25\textheight]{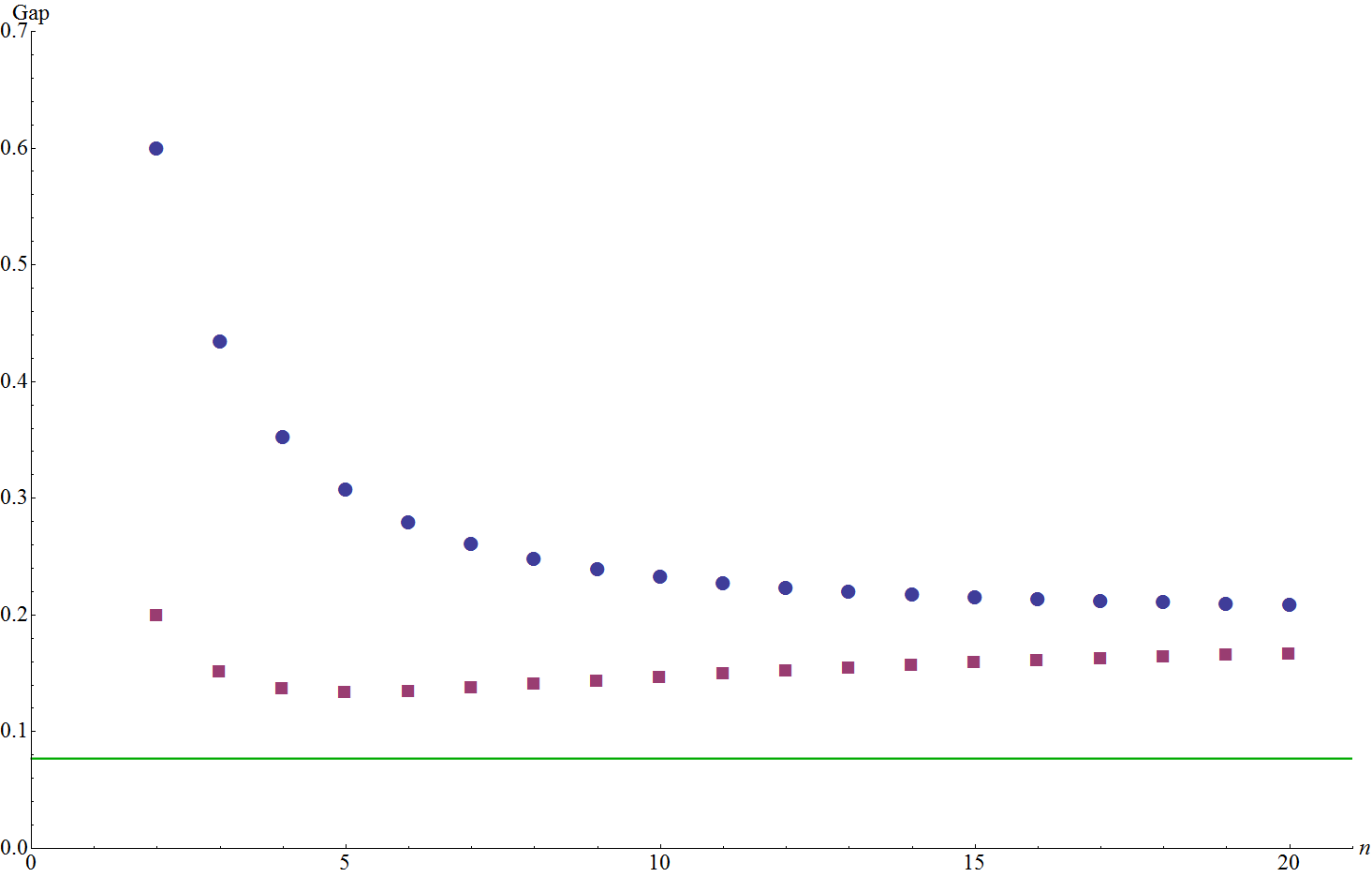}
\caption{\textbf{(Color online) Numerically evaluated spectral gaps $Gap \equiv \Delta(H_{k, t})$ for $t = 2$ (dots) vs. "local" bound for $k$ qubits (squares) vs. "global" bound for a given $n$ (here, $n=2$) (line).}
\label{fig:1}}
\end{figure}
\begin{figure}[ht]
\includegraphics[width=0.45\textwidth, height=0.25\textheight]{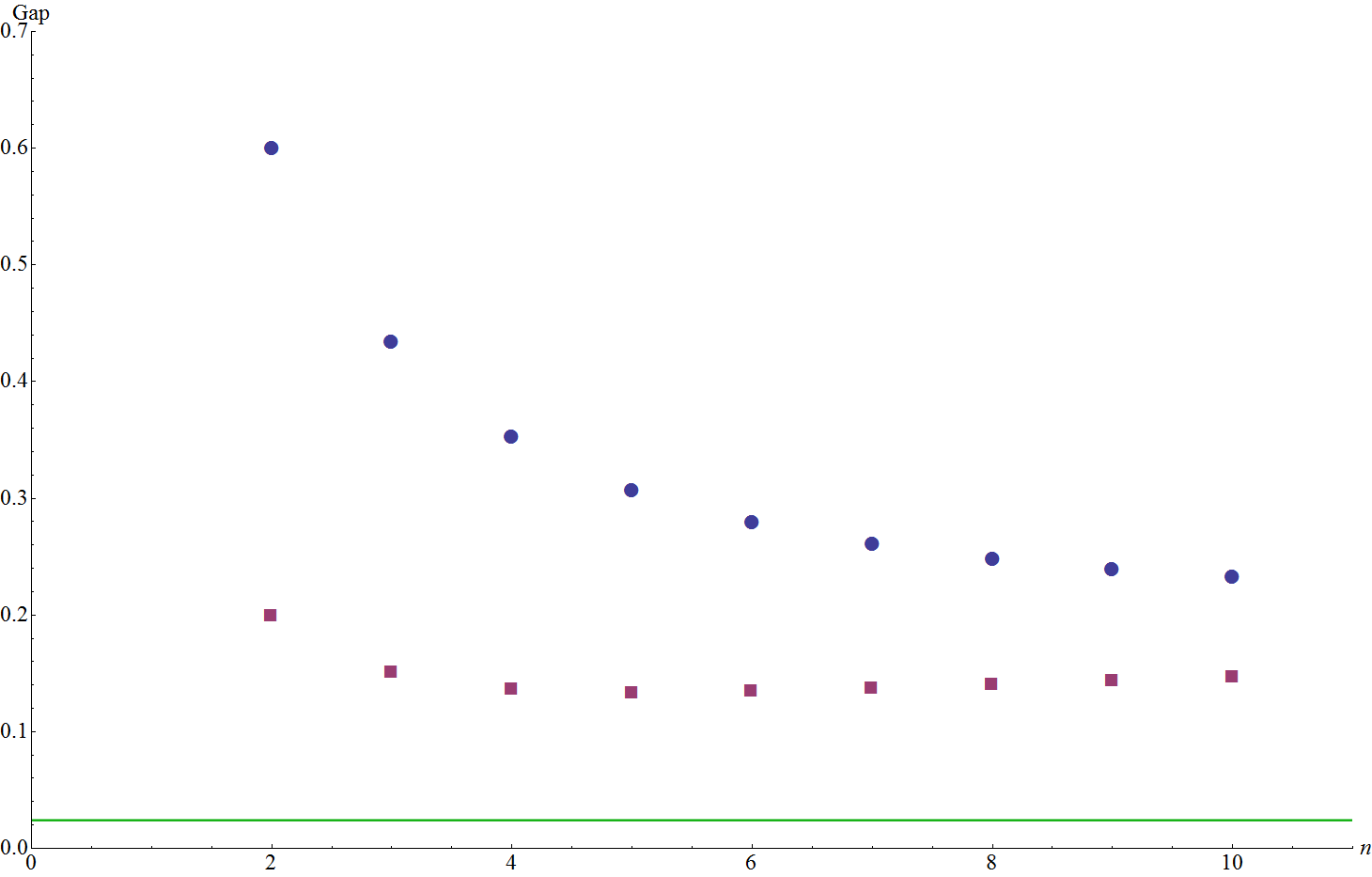}
\caption{\textbf{(Color online) Numerically evaluated spectral gaps $Gap \equiv \Delta(H_{k, t})$ for $t = 3$ (dots) vs. "local" bound for $k$ qubits (squares) vs. "global" bound for a given $n$ (here, $n=4$) (line).}
\label{fig:2}}
\end{figure}
\begin{figure}[ht]
\includegraphics[width=0.45\textwidth, height=0.25\textheight]{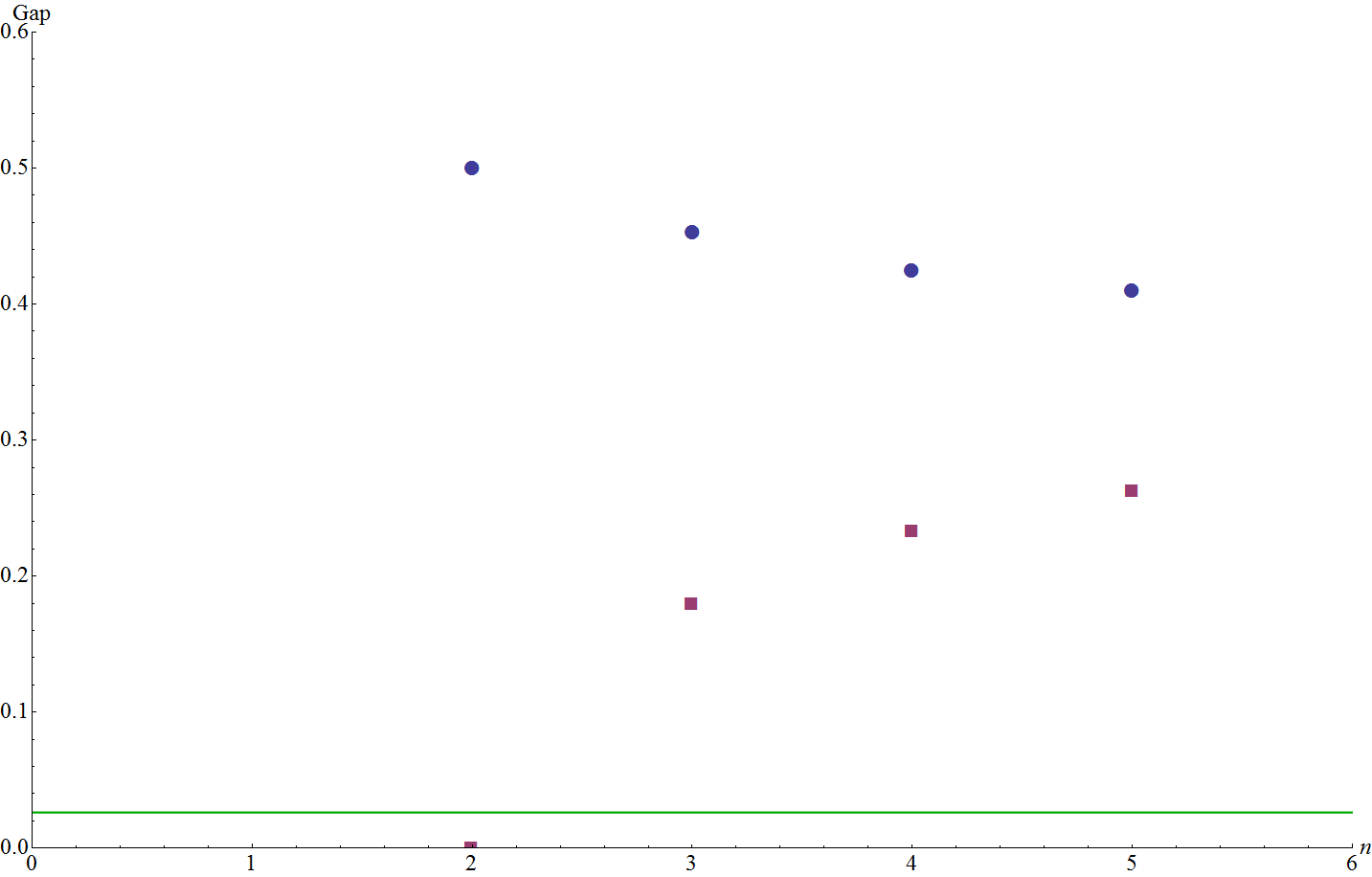}
\caption{\textbf{(Color online) Numerically evaluated spectral gaps $Gap \equiv \Delta(H_{k, t})$ for $t = 4$ (dots) vs. "local" bound for $k$ qubits (squares) vs. "global" bound for a given $n$ (here, $n=4$) (line).}
\label{fig:3}}
\end{figure}
\begin{figure}[ht]
\includegraphics[width=0.45\textwidth, height=0.25\textheight]{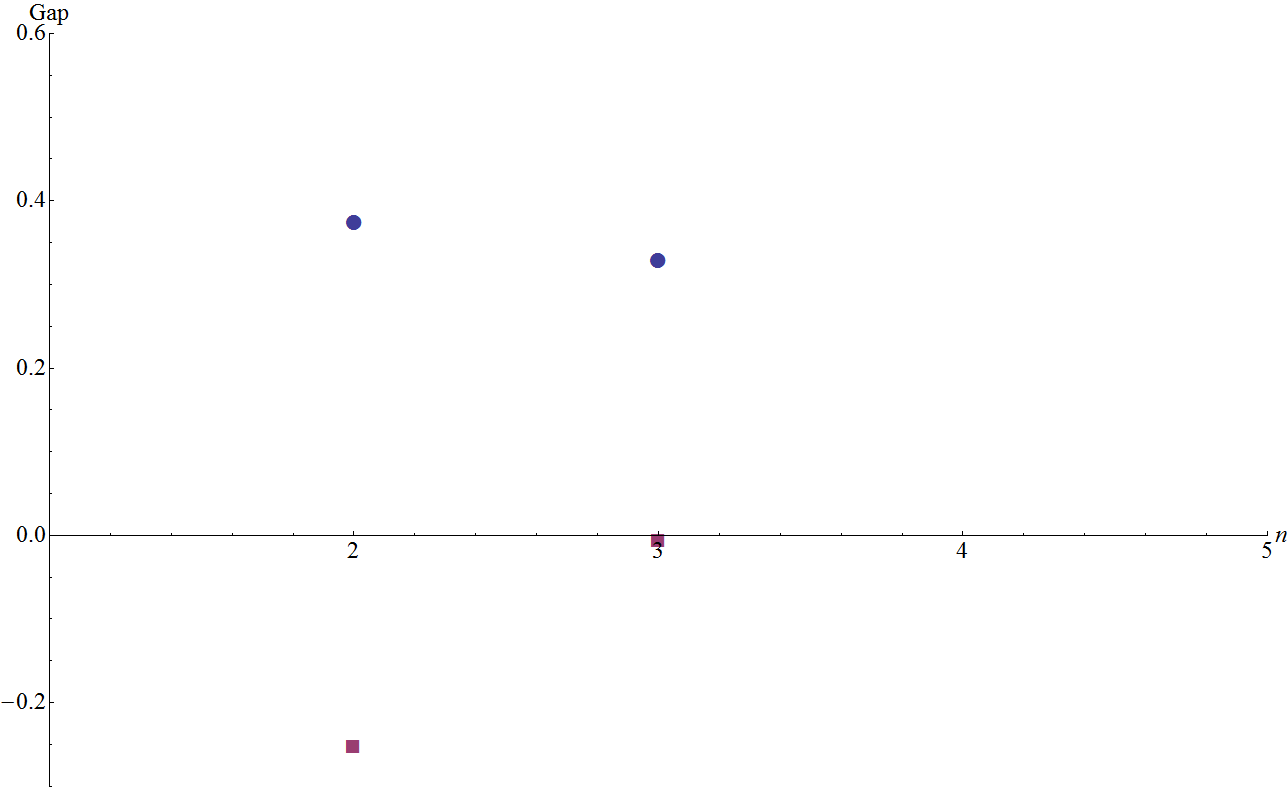}
\caption{\textbf{(Color online) Numerically evaluated spectral gaps $Gap \equiv \Delta(H_{k, t})$ for $t = 5$ (dots) vs. "local" bound for $k$ qubits (squares). Notice that the "local" one does not give any bound in this case, because it is negative. Here, the "global" bound can't be reported since a value of a gap for $H_{5,5}$ is needed to calculate it.}
\label{fig:4}}
\end{figure}

From Figures \ref{fig:1} and \ref{fig:2} one can observe that, in the case of $t = 2$ and $t = 3$, all spectral gaps are the same. This can (possibly) imply the following
\begin{conjecture}
Knowledge of the spectral gaps of the local Hamiltonian $H_{k,2}$ from Theorem \ref{th} is sufficient to know the ones for $H_{k,3}$, since there is a one to one correspondence between them.
\end{conjecture}

\begin{remark}
Let us explain here one important thing about our notation. For simplicity and transparency, we mark each lower bound from Eq. \eqref{boundKnabe} corresponding to some spectral gap as a "square" in Figures \ref{fig:1}, \ref{fig:2}, \ref{fig:3} and \ref{fig:4}.
\end{remark}
\subsection{Basic examples}
Consider the case when $t = 2$. The spectral gap for $k = 2$ is equal to $\Delta(H_{2, 2}) = \frac{6}{10}$. Using Theorem \ref{spec}, we have that, in this case, the second largest eigenvalue $\lambda_2$ of $G_{2,2}$ equals to $\frac{1}{5}$. Applying this to the "local" bound, we get that the bounded gap is equal to $\frac{2}{10}$, which implies (according to Def. \ref{def2} and Fact \ref{fact:Ginf}) that $5n\log(\frac{1}{\varepsilon})$-size random quantum circuit is an $\varepsilon$-approximate $2$-design. From the "global" bound, we get that $13.34n\log(\frac{1}{\varepsilon})$-size random quantum circuit is an $\varepsilon$-approximate $2$-design. Looking carefully at Fig. \ref{fig:1}, one can observe that $\Delta(H_{20, 2}) \approx 0.21$ seems to be bounding the convergence of the values of gaps quite well. It tells us then that the length of a circuit should scale as $4.76n\log(\frac{1}{\varepsilon})$. See the references from the introduction (especially \cite{Dankert2009-Haar}) for possible applications of these results.

When $t = 3$, $\Delta(H_{2, 3}) = \frac{6}{10}$ and thus again the above result from the "local" bound and the convergence is valid. Using the "`global"' bound we have that $35.95n\log(\frac{1}{\varepsilon})$-size random quantum circuit is an $\varepsilon$-approximate $3$-design. Note that $3$-designs can be used to solve the $U$-circuit checking problem \cite{Brandao2010-tdesign}.

The results for $t=2,3$ obtained using the "local" bound, are in accordance with these from \cite{Brandao2010-tdesign} and what is more, for $t=2$ our outcomes match those from \cite{Dankert2009-Haar}, but are a little bit more precise. 

For $4$-designs, which can be used to bound the equilibration time in some cases (see, \cite{Brandao2011-thermo, Masanes2011-thermo}), from the "local" bound, we have that a $2.44n\log(\frac{1}{\varepsilon})$-size circuit can form them. From th other hand, if the gap $\Delta(H_{5, 4}) \approx 0.41$ for $5$ qubits, which is the best explicity calculated value, can be approximately used as the bound to which all others converge then we get that $3.81n\log(\frac{1}{\varepsilon})$-size random quantum circuit is an $\varepsilon$-approximate $4$-design. Here, the "global" bound gives the following length - $37.66n\log(\frac{1}{\varepsilon})$. The result from the first bound seems to be interesting. Why?  We except that a length of a circuit should increase with increasing $t$ and comparing the results from the "local" bound, we have the circuit for $t=4$ is shorter than for $t=2,3$.

For $t=5$, not much can be said. All values calculated according to the "local" bound are negative and $\Delta(H_{2, 5}) \approx 0.29$ rather cannot be used as the value that bounds all others. Here, a prediction from the "global" bound also cannot be reported since a value of a gap for $H_{5,5}$ is needed to calculate it.

It is worth asking whether these results can be improved. As proved in \cite{Brandao-Harrow-Horodecki} (see, Proposition $8$ from that paper), neither the $t$ nor $n$ dependence can be improved by more than polynomial factors.
\section{Conclusions and Open problems}
In this paper, we numerically studied the recent statement that local random quantum circuits
acting on $n$ qubits and composed of polynomially many two-qubit gates form an approximate
unitary $\poly(n)$-design.

To this end, we evaluated spectral gaps of local Hamiltonians acting on $n$ qubits, using techniques from many-body physics and relating the degree $t$ in $t$-design to the symmetric group $S(t)$. As an additional result, it occurs that for a given $k$, $H_{k,2}$ is equal to $H_{k,3}$, while there is no such a connection between $H_{k,t}$ for higher values $t$.

What is more, we compare our results to two lower bounds for spectral gaps, leading to conclusion that for small number of qubits, lower bounding is, usually, not sufficient to obtain a reliable result - there is a big difference in actual values of spectral gaps and lower bound. For big numbers of qubits and high orders $t$, there is another problem, to obtain a "good" lower bound, one needs to calculate first, a spectral gap for a quite big value of $n$ and $t$, which is a quite complicated task, from the computational point of view. That's why, we were unable to compute the "global" lower bound for the case, when $t=5$, since it requires initial knowledge of the spectral gap for $k = 5$.

Here, one possible way, to obtain better results would be to use a super-computer and/or the power of parallel computing for calculations of spectral gaps. Another way, would be to find a "better" basis in which our operator $X$ takes a diagonal form or at least a block diagonal form. We leave these tasks as open.

At the end, we would like to point out one problem which should be of some interest, is it possible to approximate values of spectral gaps by some function with dependence on parameters $k$ and $t$? Based on the numerics, the function $(a(t) - \frac{b(t)}{k^k})^\frac{c(t)}{k}$ looks like a promising candidate.

Also, two interesting questions related to the structure of $t$-designs, namely, why all spectral gaps for $t=2$ and $t=3$ are the same and why for $t=4$ the circuit is shorter than for $t=2,3$, remain without an answer.

\section*{Acknowledgments}
We want to thank Fernando G.S.L. Brand\~ao for discussions. P.\'C. acknowledge helpful discussions with Norbert Schuch.
M.S. is supported by the International PhD Project "Physics of future quantum-based information technologies": grant MPD/2009-3/4 from Foundation for Polish Science.
The work is also supported by Polish Ministry of Science and Higher Education Grant no.
IdP2011 000361. Part of this work was done in National Quantum Information Centre of Gda\'nsk.

\section{Appendix A: Orthogonalisation of representation operators of finite groups}
\label{App:A}
In this section we we briefly remind some properties of the algebra
generated by a given complex finite dimensional representation of the finite
group $G.$ The content of this section can be found in the standard books on
representation theory of finite groups and algebras, for example, in \cite{Curtis, Littlewood}.

Any complex finite-dimensional representation $D:G\rightarrow Hom(V)$ of the
finite group $G,$ where $V$ is a complex linear space, generates a algebra $%
A_{V}[G]$ $\subset $ $Hom(V)$ which isomorphic to the group algebra $%
%TCIMACRO{\U{2102} }%
%BeginExpansion
\mathbb{C}
%EndExpansion
\lbrack G]$ if the representation $D$ is faithful. Obviously
\be
A_{V}[G]=\operatorname{span}_{%
%TCIMACRO{\U{2102} }%
%BeginExpansion
\mathbb{C}
%EndExpansion
}\{D(g),\quad g\in G\}.
\ee
If the operators $D(g)$ are linearly independent, then they form a basis of
the algebra $A_{V}[G]$ and $\dim A_{V}[G]=\left\vert G\right\vert $. It is
also possible, using matrix irreducible representations, to construct a new
basis which has remarkable properties, very useful in applications of
representation theory. Below we describe this construction.

\begin{notation}
\label{not1}
Let $G$ be a finite group of order $\left\vert G\right\vert =n$ which has $r$
classes of conjugated elements. Then $G$ has exactly $r$ inequivalent,
irreducible representations, in particular $G$ has exactly $r$ inequivalent,
irreducible matrix representations. Let
\be
D^{\alpha }:G\rightarrow Hom(V^{\alpha }),\qquad \alpha =1,2,....,r,\qquad
\dim V^{\alpha }=d_{\alpha }
\ee
be all inequivalent, irreducible representations of $G$ and let chose these
representations to be all unitary (always possible) i.e.
\be
D^{\alpha }(g)=(D_{ij}^{\alpha }(g)),\qquad \text{and} \qquad
(D_{ij}^{\alpha }(g))^{\dagger}=(D_{ij}^{\alpha }(g))^{-1},
\ee
where $i,j=1,2,....,d_{\alpha }$.
\end{notation}

The matrix elements $D_{ij}^{\alpha }(g)$ will play a crucial role in the
following.

\begin{definition}
\label{ad1}
Let $D:G\rightarrow Hom(V)$ be an unitary representation of a finite group $G
$ such that the operators $D(g),$ $\ g\in G$ are linearly independent i.e. $%
\dim A_{V}[G]=\left\vert G\right\vert $ and let $D^{\alpha }:G\rightarrow
Hom(V^{\alpha })$ be all inequivalent, irreducible representations of $G$
described in Notation~\ref{not1} above. Define
\be
E_{ij}^{\alpha }=\frac{d_{\alpha }}{n}\sum_{g\in G}D_{ji}^{\alpha
}(g^{-1})D(g),
\ee
where $\alpha =1,2,...,r,\quad i,j=1,2,..,d_{\alpha },\quad
E_{ij}^{\alpha }\in A_{V}[G]\subset Hom(V)$.
\end{definition}

The operators have noticeable properties listed in the

\begin{theorem}
\label{at1}
I) There are exactly $\left\vert G\right\vert =n$ \ \ nonzero operators $%
E_{ij}^{\alpha }$ and
\be
D(g)=\sum_{ij\alpha }D_{ij}^{\alpha }(g)E_{ij}^{\alpha }
\ee

II) the operators $E_{ij}^{\alpha }$ are orthogonal with respect to the
Hilbert-Schmidt scalar product in the space $Hom(V).$
\be
(E_{ij}^{\alpha },E_{kl}^{\beta })=\tr((E_{ij}^{\alpha })^{\dagger}E_{kl}^{\beta
})=k_{\alpha }\delta ^{\alpha \beta }\delta _{ik}\delta _{jl},\qquad
k_{\alpha }\geq 1,
\ee

\ \ where $k_{\alpha }$ \ is equal to the multiplicity of the irreducible
representation $D^{\alpha }$ in $D$ and it does not depend on $\ \ \ \ \ \ \
\ \ \ \ \ \ \ \ \ \ \ \ \ \ \ \ \ \ \ \ \ \ \ \ \ i,j=1,2,....,d_{\alpha }$

III) the operators $E_{ij}^{\alpha }$ satisfy the following composition rule%
\be
E_{ij}^{\alpha }E_{kl}^{\beta }=\delta ^{\alpha \beta }\delta
_{jk}E_{il}^{\alpha },
\ee
\ \ \ \ in particular $E_{ii}^{\alpha }$ are orthogonal projections.
\end{theorem}

\begin{remark}
\label{ar1}
From point II) of the theorem it follows that the equations.%
\be
E_{ij}^{\alpha }=\frac{d^{\alpha}}{n}\sum_{g\in G}D_{ji}^{\alpha }(g^{-1})D(g)
\ee
describe transformation of orthogonalisation of operators $D(g),$ $\ g\in G$
in the space \bigskip $Hom(V)$ with the Hilbert-Schmidt scalar product.
\end{remark}

\begin{remark}
\label{ar2}
We can look at operators $E_{ij}^{\alpha}$ as a vectors in $\mathbb{C}^{n!}$
then they are orthonormal (with a coefficient $\sqrt{\frac{d^{\alpha}}{n!}}$)  with respect to the usual scalar product in $\mathbb{C}$. So in fact we have double orthogonality.
\end{remark}

The operators  $E_{ii}^{\alpha }$ are not only orthogonal
projections onto their proper subspaces in $V$ but they are also orthogonal
with respect to the Hilbert-Schmidt scalar product in the space $Hom(V)$.

The basis $\{E_{ij}^{\alpha }\}$ play essential role when $D:G\rightarrow
%TCIMACRO{\U{2102} }%
%BeginExpansion
\mathbb{C}
%EndExpansion
\lbrack G]$ is the regular representation. In this case the properties of the
basis $\{E_{ij}^{\alpha }\}$ express the well-known fact that the group
algebra $%
%TCIMACRO{\U{2102} }%
%BeginExpansion
\mathbb{C}
%EndExpansion
\lbrack G]$ is a direct sum of simple matrix algebras generated by the
irreducible representations of the group $G$. It is always possible to
construct the operators $E_{ij}^{\alpha }$ even if the operators $D(g)$ are
not linearly independent but in this case some of them will be zero.

From Theorem~\ref{at1} it follows directly
\begin{corollary}
\begin{widetext}
\begin{eqnarray}
E_{ij}^{\delta (12)} =\frac{d_{\delta }}{n}\sum_{g\in G}D_{ji}^{\delta
}(g^{-1})D(g)\otimes D(g)=
\sum_{\alpha \beta }\sum_{kl}\sum_{mn}\left( \frac{d_{\delta }}{n}%
\sum_{g\in G}D_{ji}^{\delta }(g^{-1})D_{lk}^{\alpha }(g)D_{nm}^{\beta
}(g)\right) E_{kl}^{\alpha }\otimes E_{mn}^{\beta},
\end{eqnarray}
\end{widetext}
thus the coefficients on $RHS$ are expressed only by the matrix elements of
irreducible representations.
\end{corollary}

\section{Appendix B: The largest eigenvalue of $G_{\mu,t}$}
\label{App:B}
Let us recall one more time the fact present in Theorem \ref{th} (after a slight modification).
\begin{fact}
We have
\be ||G_{\mu^{*l},t}-G_{\mu_H,t}||_{\infty} = \lambda_2^l, \ee
where $\lambda_2$ is the second largest eigenvalue of $G_{\mu,t}$. Moreover the largest
eigenvalue $\lambda_1$ of $G_{\mu,t}$ is equal to $1$, and the corresponding eigenprojector
is equal to $G_{\mu_H,t}$.
\end{fact}
Now, we will prove that the largest eigenvalue of $G_{\mu,t}$ is equal to $1$.
\begin{proof}
We know that $G_{\mu,t}$ is a operator such that $0 \leq G_{\mu,t} \leq I$. Having this in mind, let us eigen-decompose operator $G_{\mu,t}$ in some basis as
\be G_{\mu,t} = \frac{1}{n} \sum_{i=1}^{n} P_{i,i+1} = \sum_w v_w Q_w, \label{eig} \ee
where $v_w$ are eigenvalues, so they follow the constraint $0 \leq v_w \leq 1$, and $Q_w$ are the corresponding eigenvectors. Now, let us extend Eq. \eqref{eig} to $l$ walks in the random walk:
\be G_{\mu^{*l},t} = (\sum_w v_w Q_w)^l = \sum_w v_w^l Q_w, \ee
since $Q_w$ are projectors. We also have the following
\be \mathop{\operatorname{lim}}\limits_{l \rightarrow \infty} G_{\mu^{*l},t} = G_{\mu_H,t}. \ee
Then, it is easy to observe that the correspondence is valid only when $v_w = 1$ and the eigenprojector $Q_w$ is equal to $G_{\mu_H,t}$.
\end{proof}

\section{Appendix C: Superoperators and operators}
\label{App:C}
Here, we explore some connections between superoperators and operators. For a superoperator ${\cal G}$ given by
\be {\cal G}(X) = \sum_{k} A_k X B_{k}^{\dag}, \label{eq1} \ee
where $\dag$ denotes the Hermitian conjugate; we define the operator
\be G = \sum_k A_k \ot B_{k}^{\star}, \label{G}\ee
with $\star$ being the complex conjugate.

Let $X$ be a normalized operator, such that $\tr(X X^{\dag}) = 1$. What is more, assume that ${\cal G}(X) = \lambda X$, for a complex eigenvalue $\lambda$, i.e., $X$ is an eigenoperator of ${\cal G}$ with eigenvalue $\lambda$. Then defining $|X\> = X \ot \II |\phi\>$, with
\be |\phi\> = \sum_k |k\> \ot |k\>, \ee
it holds that $G|X\> = \lambda |X\>$, i.e., $|X\>$ is an eigenvector of $G$ with eigenvalue $\lambda$.

A direct implication of this correspondence is that
\begin{fact}
\label{fact:Ginf}
\be ||{\cal G}||_{2 \rightarrow 2} = ||G||_{\infty}, \ee
where the norm $2 \rightarrow 2$ is defined as in Def. \ref{def2}.
\end{fact}
\begin{proof}
Using Eq. \eqref{eq1} and \eqref{G} we have that
\be ||\Lambda(X)||_{2\rightarrow2} = sup\frac{|| \Lambda(X) ||_2}{||X||_2}. \ee

Decomposing $X = \sum_k c_k X_k$, $\Lambda(X_k) = \lambda_k X_k$ and noting that $\sum_k |c_k|^2 = 1$,
\be
||\Lambda||_{2\rightarrow2}^2 = sup \frac{\sum_k |c_k|^2 |\lambda_k|^2 X_k^2}{\sum_k c_k^2 X_k^2} = \lambda_{max}^2.
\label{proof}
\ee
Recalling the definition of the norm $||.||_{\infty}$, one can observe that Eq. \eqref{proof} proofs Fact \ref{fact:Ginf}.
\end{proof}

\bibliographystyle{apsrev}
\bibliography{rmp15-hugekey}
\end{document}